\documentclass[10pt,twocolumn,journal]{IEEEtran}

\usepackage{amsmath, mathrsfs, mathtools}
\usepackage{amsfonts}
\usepackage{amssymb}
\usepackage{color}
\usepackage{multirow}
\usepackage{stmaryrd}
\usepackage{mathabx}

\allowdisplaybreaks

\usepackage{caption}
\usepackage{subcaption}
\usepackage{float}

\usepackage{amsfonts}
\usepackage{amssymb}
\usepackage{color}
\usepackage{multirow}
\usepackage{graphicx}
\usepackage{tcolorbox}

\newcommand{\subparagraph}{}
\usepackage[compact]{titlesec}
\titlespacing*{\section}{15pt}{1.2\baselineskip}{0.9\baselineskip}

\allowdisplaybreaks

\newcommand{\myhash}{%
	{\settoheight{\dimen0}{C}\kern-.05em\, \resizebox{!}{\dimen0}{\raisebox{\depth}{\#}}}}

\usepackage[numbers,sort&compress]{natbib}

\newcommand{\Sigmay}{{\Sigmay}_{\yv}}

\def\lambdam{{\boldsymbol{\lambda}}}

\def\phim{\boldsymbol{\phi}}

\usepackage{multicol}

\usepackage{algorithm}
\usepackage{algpseudocode}
\usepackage{pifont}
\usepackage{varwidth}


\usepackage{mathbbol}

\def\mindex#1{\index{#1}}

\def\sq{\hbox{\rlap{$\sqcap$}$\sqcup$}}
\def\qed{\ifmmode\sq\else{\unskip\nobreak\hfil
\penalty50\hskip1em\null\nobreak\hfil\sq
\parfillskip=0pt\finalhyphendemerits=0\endgraf}\fi\medskip}

\long\def\defbox#1{\framebox[.9\hsize][c]{\parbox{.85\hsize}{%
\parindent=0pt
\baselineskip=12pt plus .1pt      
\parskip=6pt plus 1.5pt minus 1pt 
 #1}}}

\long\def\beginbox#1\endbox{\subsection*{}%
\hbox{\hspace{.05\hsize}\defbox{\medskip#1\bigskip}}%
\subsection*{}}

\def\endbox{}

\newsavebox{\junk}
\savebox{\junk}[1.6mm]{\hbox{$|\!|\!|$}}

\def\argmin{\mathop{\rm arg\, min}}

\newcommand{\field}[1]{\mathbb{#1}}

\def\Re{\field{R}}

\def\ind{\field{I}}

\def\bC{{\mathbb C}}

\def\bR{{\mathbb R}}

\def\bfA{{\bf A}}

\def\bfD{{\bf D}}

\def\bfI{{\bf I}}

\def\bfW{{\bf W}}

\def\bfa{{\bf a}}
\def\bfb{{\bf b}}

\def\bfe{{\bf e}}

\def\bfh{{\bf h}}

\def\bft{{\bf t}}

\def\bfy{{\bf y}}
\def\bfz{{\bf z}}

\def\scrC{{\mathscr{C}}}

\def\sfF{{\sf F}}

\def\sfH{{\sf H}}

\def\bfmath#1{{\mathchoice{\mbox{\boldmath$#1$}}%
{\mbox{\boldmath$#1$}}%
{\mbox{\boldmath$\scriptstyle#1$}}%
{\mbox{\boldmath$\scriptscriptstyle#1$}}}}

\def\bfmY{\bfmath{Y}}

\def\bfmhhaY{\bfmath{\hhaY}} 
\def\bfmhhaY{\hbox to 0pt{$\widehat{\bfmY}$\hss}\widehat{\phantom{\raise 1.25pt\hbox{$\bfmY$}}}}

\def\til={{\widetilde =}}

\def\clC{{\cal C}}
\def\clD{{\cal D}}

\def\clN{{\cal N}}

\def\clP{{\cal P}}

\def\clS{{\cal S}}

 \def\FRAC#1#2#3{\genfrac{}{}{}{#1}{#2}{#3}}

\def\ddtp{{\mathchoice{\FRAC{1}{d^{\hbox to 2pt{\rm\tiny +\hss}}}{dt}}%
{\FRAC{1}{d^{\hbox to 2pt{\rm\tiny +\hss}}}{dt}}%
{\FRAC{3}{d^{\hbox to 2pt{\rm\tiny +\hss}}}{dt}}%
{\FRAC{3}{d^{\hbox to 2pt{\rm\tiny +\hss}}}{dt}}}}

\def\average#1,#2,{{1\over #2} \sum_{#1}^{#2}}

\def\eye(#1){{\bf(#1)}\quad}

\newtheorem{proposition}{{\bf Proposition}}

\def\eq#1/{(\ref{e:#1})}

\newcommand{\beqn}[1]{\notes{#1}%
\begin{eqnarray} \elabel{#1}}

\newcommand{\eeqn}{\end{eqnarray} }

\newcommand{\beq}[1]{\notes{#1}%
\begin{equation}\elabel{#1}}

\newcommand{\eeq}{\end{equation}}

\def\bdes{\begin{description}}
\def\edes{\end{description}}

\newcounter{rmnum}

\newcounter{anum}

{\end{list}}

\def\ass(#1:#2){(#1\ref{#1:#2})}

\def\ritem#1{
\item[{\sf \ass(\current_model:#1)}]
}

\newenvironment{recall-ass}[1]{%
\begin{description}
\def\current_model{#1}}{
\end{description}
}

\usepackage{tikz}
\usepackage{pgfplots,tikz-3dplot}
\usetikzlibrary{intersections, pgfplots.fillbetween}
\pgfplotsset{compat=newest}
\usetikzlibrary{patterns}
\usetikzlibrary{positioning}
\usetikzlibrary{datavisualization}
\usetikzlibrary{datavisualization.formats.functions}
\usetikzlibrary{backgrounds}
\usetikzlibrary{shapes,snakes}

\usepgfplotslibrary{external} 
\usepackage{float}
\usepackage{afterpage}
\usepackage{balance}

\def\herm{{\sfH}}

\def\sigmam{{\boldsymbol{\sigma}}}

\newcommand{\normd}[1]{{\left\vert\kern-0.25ex\left\vert\kern-0.25ex\left\vert #1 
		\right\vert\kern-0.25ex\right\vert\kern-0.25ex\right\vert}}

\long\def\comment#1{}

\newcommand{\yv}{{\bf y}}

\newcommand{\Sigmam}{\hbox{\boldmath$\Sigma$}}
\newcommand{\Phim}{\hbox{\boldmath$\Phi$}}

\newcommand{\Psim}{\hbox{\boldmath$\Psi$}}

\renewcommand{\Re}{{\rm Re}}

\newcommand{\transp}{{\sf T}}

\def\gammam{\boldsymbol{\gamma}}

\title{Machine Learning for Geometrically-Consistent Angular Spread Function Estimation in Massive MIMO}

\author{Yi Song,  Mahdi Barzegar Khalilsarai, Saeid Haghighatshoar,   and Giuseppe Caire  
\thanks{The authors are with the Communications and Information Theory Group (CommIT), Technische Universit\"{a}t Berlin (\{yi.song, m.barzegarkhalilsarai, saeid.haghighatshoar,  caire\}@tu-berlin.de).}
}

\begin{document}

\maketitle

\newpage

\def\ful{f_\text{ul}}
\def\fdl{f_\text{dl}}
\def\asfc{\scrC}
\def\asful{\scrC_\text{ul}}
\def\asfdl{\scrC_\text{dl}}

%
%

\begin{abstract}
	In the spatial channel models used in multi-antenna wireless communications, the propagation from a single-antenna transmitter (e.g., a user) to an $M$-antenna receiver (e.g., a Base Station) 
	occurs through scattering clusters located in the far field of the receiving antenna array.  
	The Angular Spread Function (ASF) of the corresponding $M$-dim channel vector  describes 
	the angular density of the received signal power at the array. 
	The modern literature on massive MIMO
	has recognized that the knowledge of  covariance matrix of  user channel vectors is very useful 
	for various applications such as hybrid digital analog beamforming, pilot decontamination,  etc.
	Therefore, most literature has focused on the estimation of such channel covariance matrices.  
	However, in some applications such as uplink-downlink covariance transformation (for FDD massive MIMO precoding) and channel sounding some form of ASF estimation is required either implicitly or explicitly. 
	 It turns out that while covariance estimation is well-known and well-conditioned, 	the ASF estimation is a much harder problem and  is in general ill-posed. 
	In this paper, we show that under additional \textit{geometrically-consistent group-sparsity} structure on the ASF, which is prevalent  in almost all wireless propagation scenarios, one is able to estimate ASF properly. 
	We  propose  sparse dictionary-based algorithms that promote this group-sparsity structure via suitable regularizations. Since generally it is difficult to capture the notion of group-sparsity through proper regularization, we propose another algorithm based on \textit{Deep Neural Networks} (DNNs) that learns this structure. 
	We provide numerical simulations to assess the performance of our proposed algorithms. We also compare the results with that of other methods  in the literature, where we re-frame those methods in the context of ASF estimation in massive MIMO. 
	
\end{abstract}


\section{Introduction}
Consider a massive MIMO system \cite{Marzetta-TWC10} with a BS with $M\gg 1$ antennas serving several single-antenna users. We assume that BS antennas lie on a \textit{Uniform Linear Array} (ULA) with standard antenna spacing $d=\frac{\lambda}{2}$ where $\lambda=\frac{c_0}{f_0}$ denotes the wavelength with $c_0$ and $f_0$ being the speed of light and the carrier frequency, respectively.
We consider a generic user and assume that the propagation between this user and the Base Station (BS) array occurs through a collection of sparse scatters (such as buildings, trees, etc.) in the \textit{Angle-of-Arrival} (AoA) domain as illustrated in Fig.\,\ref{fig:scat_chan}.
We consider a  block-fading model, widely-adopted as a wireless channel model \cite{tse2005fundamentals}, where the channel vector of a user at a specific resource block $s$ is given by 
\begin{align}
\bfh(s)=\sum_{i=1}^k w_i(s) \bfa(\xi_i),
\end{align}
where $\{w_i(s): i\in [k]\}$ and $\{\xi_i: i \in [k]\}$ denote the random channel coefficients and the AoAs of the $k$ scatterers in the channel, respectively, and where $\bfa(\xi)$ denotes the array response vector at the  AoA parametrized by $\xi=\sin(\theta)$ (in terms of the AoA $\theta$), which for the standard array spacing is given by $\bfa(\xi)=(1, e^{j \pi \xi}, e^{j 2 \pi \xi}, \dots, e^{j \pi (M-1)\xi})^\transp$, where $M$ denotes the number of BS antenna as before. 
\begin{figure}[t]
	\centering
	\includegraphics[scale=0.35]{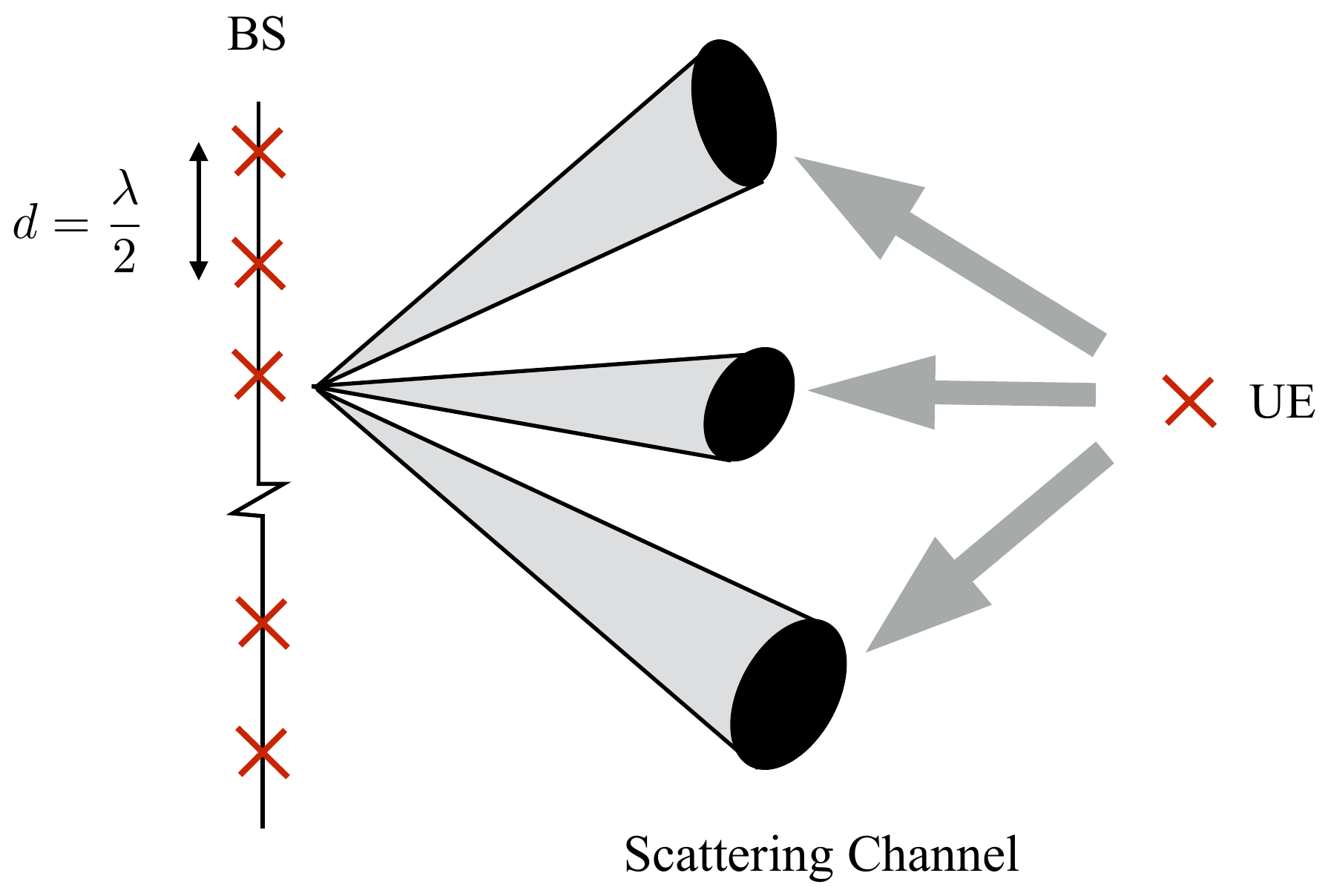}
	\caption{Sparse scattering channel between a generic user and the BS. In this example, the channel consists of 3 large scatterers reflecting the power of the user to the BS array.}
	\label{fig:scat_chan}
\end{figure}
We assume that the channel coefficients have a Gaussian distribution $w_i(s)\sim \clC\clN (0, \gamma_i)$ where $\gamma_i$ denotes the channel strength of the $i$-th scatterer. Due to the Gaussian assumption, the  statistics of the channel  can be fully specified by the covariance matrix of the channel vector $\bfh(s)$ given by
\begin{align}\label{disc_model}
\Sigmam=\sum_{i=1}^k \gamma_i \bfa(\xi_i) \bfa(\xi_i)^\herm.
\end{align}
The channel model \eqref{disc_model} is more suitable for an ideal scenario when the scatterers are specular with very narrow AoA width. In practice, however, the scatterers are physical object with finite but non-zero width, thus, it is more realistic to consider  a diffuse scattering  given by the continuous limit of \eqref{disc_model} with
\begin{align}\label{eq:asf_cov}
\Sigmam=\int _{-1}^{1} \gamma(\xi) \bfa(\xi) \bfa(\xi)^\herm d \xi
\end{align}
where $\gamma(\xi)$ is a positive function representing the density of the power of scatterer seen at a narrow AoA range $[\xi, \xi+d \xi]$, which we call the \textit{Angular Spread Function} (ASF) of the channel vector $\bfh(s)$. The specular case in \eqref{disc_model} can be obtained from this general model by approximating each specular component $(\gamma_i, \xi_i)$ with a narrow rectangular pulse centered at $\xi=\xi_i$ and with an amplitude proportional to $\gamma_i$.

The modern literature on massive MIMO 
has recognized that the knowledge of  covariance matrix of  user channel vectors as in \eqref{eq:asf_cov} is very useful for various applications in massive MIMO such as user grouping and scheduling \cite{xie2016unified}, hybrid digital analog beamforming  \cite{adhikary2013joint, haghighatshoar2016massive, haghighatshoar2018low}, pilot decontamination \cite{yin2013decontaminating, haghighatshoar2017massive, yin2015pilot, bjornson2016pilot}, and low-complexity beamforming \cite{benzin2019low}.
Therefore, most literature has focused on the estimation of such channel covariance matrices.  
However, in some applications such as uplink-downlink covariance transformation for FDD massive MIMO precoding \cite{khalilsarai2018fdd, miretti2018fdd}  some form of ASF estimation is required either implicitly or explicitly. Also, ASF estimation allows channel sounding with an antenna array rather than with mechanically steered horn antennas as often done in nowadays practical channel measurement, thus, it yields faster and more precise channel measurements. 
It turns out that while covariance estimation is well-known and well-conditioned, 	the ASF estimation is a much harder problem and  is in general ill-posed. 
In this paper, our goal is to estimate the ASF $\gamma$ from a collection of  i.i.d. channel vectors $\{\bfh(s): s \in [T]\}$, where $T$ denotes the number of the samples. These i.i.d. data can be collected by sampling the channel vectors of the user at resource  blocks sufficiently separated in time, frequency, or both (separation measured in terms of the coherence time and coherence bandwidth of the channel).  We denote the sample covariance of the samples $\{\bfh(s): s\in [T]\}$ by
 \begin{align}
\widehat{\Sigmam}_\bfh=\frac{1}{T} \sum_{s\in [T]} \bfh(s) \bfh(s)^\herm.
\end{align}
Note that, for Gaussian channel vectors, $\widehat{\Sigmam}$ is a sufficient statistics for $\Sigmam$ and is the  only information available for ASF estimation. Unfortunately,  even in the ideal case where $T \to \infty$ and $\widehat{\Sigmam}$ tends to true covariance matrix $\Sigmam$ in \eqref{eq:asf_cov}, the estimation of ASF is ill-posed. The main reason is that \eqref{eq:asf_cov} is a mapping between the infinite-dim set of ASF to the finite-dim set of covariance matrices, thus, there maybe  many ASFs corresponding to the same covariance matrix $\Sigmam$. We can see this easily in the specific case of ULA, where $\Sigmam$ turns out to be a \textit{positive semi-definite} (PSD) Toeplitz matrix whose first column is given by $\sigmam=(\sigma_0, \dots, \sigma_{M-1})^\transp \in \bC^M$ with
\begin{align}\label{eq:fourier}
\sigma_r=\int_{-1}^1 \gamma(\xi) e^{j \pi r \xi} d\xi.
\end{align}
One can see from \eqref{eq:fourier} that $\sigmam$ contains only the first $M$ Fourier coefficients of $\gamma(\xi)$. Of course, in the asymptotic scenario where the number of BS antennas $M \to \infty$, one would be able to recover the ASF $\gamma$ from all its Fourier coefficients in \eqref{eq:fourier}, thus,  from the covariance matrix $\Sigmam$. However, for any finite $M$, as is the case in all massive MIMO implementation, the recovery of the  ASF from the covariance matrix is  ill-posed unless one imposes additional structure on the set of  ASFs $\gamma(\xi)$. 

\subsection{Contribution}
In this paper, we propose a novel method for ASF estimation in massive MIMO. To make the recovery feasible, we impose  additional structure  on the set of ASFs by assuming that the ASFs are \textit{group-sparse} in  the angular domain. More precisely, as illustrated in Fig.\,\ref{fig:scat_chan}, the group-sparsity implies that a generic ASF $\gamma(\xi)$ can be decomposed as 
\begin{align}\label{asf_group}
\gamma(\xi)=\sum_{k=1}^K \gamma_k(\xi) \ind\{\xi \in \clS_k\}
\end{align}
where $\ind\{.\} \in \{0,1\}$ denotes the indicator function, where $\{\clS_k: k \in [K]\}$ is a group of $K$ mutually disjoint support sets  in the angular domain, and where $\gamma_k$ is the fraction of $\gamma$ supported on $\clS_k$ corresponding to $K$ scatterers. Moreover, we assume that each support set $\clS_k$ is a connected set, and the whole support $\clS=\cup_{k=1}^K \clS_k$ is a much smaller subset of the  set of all feasible AoAs. 
Note that, by assuming that each $\clS_k$ is a connected set, we aim to promote the group-sparsity of the ASF $\gamma$, which  is motivated by the fact that $\clS_k$ corresponds to the reflection area of a scatterer as seen from the BS antenna perspective, thus, it is a connected set (see, e.g., Fig.\,\ref{fig:scat_chan}).
For a ULA, studied in this paper, with the set of AoAs $[-1,1]$, the set of supports $\clS_k$, $k \in [K]$, correspond to a collection of $K$ non-overlapping intervals (connected sets in $[-1,1]$) of the form $[\xi_k^\text{i}, \xi_k^\text{e}]$. Fig.\,\ref{fig:scat_chan} illustrates an example of group-sparse scattering channel with $K=3$ scatterers. 

%

In this paper, we first propose a \textit{Non-Negative Least Squares} (NNLS) algorithm for recovering the ASF from noisy samples. We show that although this algorithm promotes the sparsity of the ASF in the angle domain it is unable to promote the group-sparsity. Then, we modify NNLS by introducing a new type of regularization, which we prove to promote the group-sparsity of the estimated ASF. As an alternative approach, we use  \textit{Deep Neural Networks} (DNNs) \cite{goodfellow2016deep} and train them using group-sparse ASFs. Our results show that interestingly DNNs are powerful enough to extract the ``group-sparsity'' structure from the training data, and alleviate the need for any additional regularization. Moreover, DNNs are also quite fast in computation since they do not require running time-consuming iterative optimization methods needed for regularization-based methods. We perform numerical simulations to compare the performance of our proposed methods with that of other competitive methods in the literature after re-framing them in the context of ASF estimation in massive MIMO.

\subsection{Related Work}
Interestingly, ASF estimation for  the special case of ULA boils down to the well-known classical spectral estimation problem \cite{kay1999modern,stoica2005spectral} where the goal is to estimate the power spectral density of a scalar stationary process from its $M$ time samples $\{y_t: t\in [M]\}$ with the following two differences:
\begin{itemize}
	\item rather than time samples, one has access to the samples along the antennas given by the channel vector $\bfh= \{h_i: i \in [M]\}$ where $h_i$ denotes the sample an the antenna $i$.
	
	\item one has access to multiple (with our notation $T$) i.i.d. realization of the process, namely, $\bfh(s)=\{h_i(s): i \in [M]\}$ for $s\in [T]$, rather than the single realization $\{y_t: t \in [M]\}$ typically encountered in the spectral estimation in the  classical scenario; this facilitates the ASF estimation in ULA further.
\end{itemize}
Therefore, at least in theory, one can apply all the classical methods such as periodogram for the ASF estimation in ULA; we refer to \cite{kay1999modern} for the vast literature on classical spectral analysis and to \cite{stoica2005spectral} for more modern techniques. However, as we illustrate using numerical simulations, these methods do not suit for capturing  sparsity, and in particular group-sparsity of the ASF, we address in this paper.  
A recent work that  studies indirectly the ASF estimation for massive MIMO is \cite{miretti2018fdd}, which applies $\ell_2$-norm minimization  followed by iterative projection in Hilbert spaces to solve the following optimization problem 
\begin{align}\label{HHI_method}
\widehat{\gamma}(\xi)= \argmin _{\mu\geq 0} \|\mu\|_2 \text{ s.t. } \int_{-1}^1 \mu(\xi) e^{j \pi r \xi} d \xi=\sigma_r,
\end{align}
where $\|\mu\|_2=\sqrt{\int _{-1}^{1} \mu(\xi)^2 d\xi}$ denotes the $\ell_2$-norm of $\mu$, and where $\sigma_r$ are Fourier samples of the original ASF as in \eqref{eq:fourier}.
As we illustrate in the simulations, also well-known in the literature, $\ell_2$-norm minimization does not promote the sparsity of the ASF, thus, it produces anti-sparse rather than sparse ASFs.

Recovery of signals under group-sparsity addressed here is also widely studied in the recent Compressed Sensing (CS) literature (see, e.g., \cite{baraniuk2010model, eldar2009block} and refs. therein). However, all the proposed methods  exploit  group-sparsity across multiple signal samples, known as Multiple Measurement Vector (MMV) problem, or for a single signal sample but assuming that the support of each group in the signal is explicitly known \cite{baraniuk2010model, eldar2009block}.  In our case, we estimate only a single ASF, so the MMV setting is not applicable.  Moreover, the support and size of each group is given by the AoA range and width of the corresponding  scatter, which is a priori unknown and also changes from one scatterer to the other. As a result, the proposed CS methods are ineffective for capturing the group-sparsity we study here. To the best of our knowledge, the problem of signal recovery under the most general group-sparsity structure we study here has not been previously studied in the literature. Fortunately, we are able to tackle this problem by introducing a novel regularization technique. The key to the success of our method lies in the positivity of the ASF.

As explained before, another method we use for ASF estimation is based on DNNs \cite{goodfellow2016deep}. 
DNNs have recently created a revolution in Machine Learning (ML) community and have provided a new paradigm for how ML can be adopted in a wide variety of real-world problems. Related to the wireless applications we are interested in this paper, DNNs have been applied in  many wireless communication applications such as signal detection \cite{ye2018power}, channel encoding \cite{cammerer2018end, farsad2018deep}, decoding \cite{gruber2017deep, nachmani2018deep}, signal estimation \cite{he2018deep}, and resource allocation \cite{sun2018learning}. In this paper, we use DNNs as an alternative to the group-sparsity promoting regularization we already discussed. More specifically, we train DNN such that it is able to learn the notion of \textit{group-sparsity} from training samples, where afterwards we use it as a black-box algorithm that is able to estimate the group-sparse ASF from the observation of noisy channel vectors samples.  This alleviates  need for any group-sparsity regularization. Moreover, it has the fundamental advantage that one does not need to run any iterative algorithm, which typically requires many iteration to converge; instead one immediately computes the ASF estimate by straightforward calculations through the network, which can be done even in parallel to obtain a tremendous speed-up.


\newcommand{\ch}[1]{\check{#1}}
\newcommand{\wch}[1]{\widetilde{#1}}

\section{Proposed ASF Estimation Algorithms}
\subsection{Basic Setup}
In this section, we introduced our proposed algorithms for ASF estimation. We assume that we have a collection of $T$ i.i.d. noisy channel vectors $\{\bfy(s): s\in [T]\}$, where $\bfy(s)=\bfh(s)+\bfz(s)$ where $\bfz(s)\sim \clC\clN(0, N_0 \bfI)$ is  the measurement noise and where $\bfh(s)$ is the channel vector produced  by a group-sparse ASF $\gamma(\xi)$. We denote the sample covariance of the noisy channel vectors by $\widehat{\Sigmam}_\bfy=\frac{1}{T} \sum_{s\in [T]} \bfy(s) \bfy(s)^\herm$.

For the algorithm proposed in this section, we first approximate the ASF $\gamma(\xi)$ with a dictionary of rectangular pulses  
\begin{align}\label{dic_est}
\wch{\gamma}(\xi)=\sum_{g=1}^G \gamma_g \psi_g(\xi)
\end{align}
where $\psi_g(\xi)=R(\xi-\xi_g)$  with $R(\xi)$ being a rectangular pulse of width $\frac{2}{G}$ centered at $0$, and where $\xi_g=-1+\frac{2(g-1)}{G}$ belongs to the uniform grid of size $G$ over the set of AoAs $[-1,1]$. Fig.\,\ref{fig:rect_dict} illustrates this dictionary. 

\begin{figure}[t]
	\centering
	\includegraphics[scale=0.7]{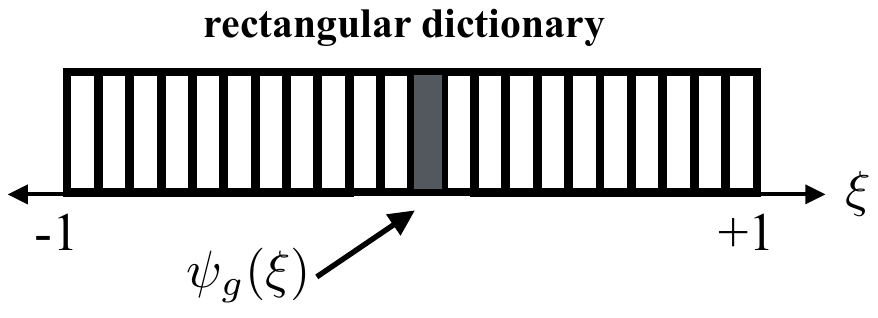}
	\caption{A dictionary for ASF approximation consisting of adjacent rectangular pulses. }
	\label{fig:rect_dict}
\end{figure}

The covariance matrix corresponding to the approximation \eqref{dic_est} is given by
\begin{align}\label{dic_mat_est}
\wch{\Sigmam}(\gammam)&=\int _{-1} ^ 1 \sum_{g=1}^G \gamma_g \psi_g(\xi) \bfa(\xi) \bfa(\xi^\herm)\nonumber\\
&=\sum_{g=1}^G \gamma_g \int_{-1}^1 \psi_g(\xi) \bfa(\xi) \bfa(\xi)^\herm d\xi\nonumber\\
&=\sum_{g=1}^G \gamma_g \Psim_g,
\end{align}
where $\Psim_g=\int_{-1}^1 \psi_g(\xi) \bfa(\xi) \bfa(\xi)^\herm d\xi$ is the covariance matrix corresponding to the rectangular pulse $\psi_g(\xi)$. Note that for the ULA, all the matrices $\Psim_g$ are \textit{positive semi-definite} (PSD)  Toeplitz matrices.

\def\psim{\boldsymbol{\psi}}
\subsection{Non-Negative Least Squares}
To estimate the coefficients $\gammam=(\gamma_1, \dots, \gamma_G)^\transp$ in \eqref{dic_est}, we use the following optimization problem
\begin{align}\label{eq:nnls_1}
\gammam^\star=\argmin _{\gammam \geq 0} \left \|\widehat{\Sigmam} - \wch{\Sigmam}(\gammam)-N_0 \bfI\right \|_\sfF^2,
\end{align}
where $\wch{\Sigmam}(\gammam)=\sum_{g=1}^G \gamma_g \Psim_g$ is given by \eqref{dic_mat_est}, where $N_0$ denotes the noise power, and where we assumed that all the coefficients of $\gammam$ are positive.Since all the matrices $\Psim_g$ are Hermitian Toeplitz, denoting by $\psim_g$ the first column of $\Psim_g$, we can write \eqref{eq:nnls_1}  equivalently as 
\begin{align}\label{eq:nnls_2}
\gammam^\star=\argmin _{\gammam \geq 0} \left \|\bfW \widehat{\sigmam} - \sum_{g=1}^G \gamma_g \bfW  {\psim}_g  - N_0[\bfW]_{1,1} \bfe_1\right \|^2,
\end{align}
where $\bfe_1=(1,0,\dots,0)$ denotes the first canonical vector, where $\widehat{\sigmam}$ is the first column of the matrix obtained by the Toeplizification (averaging over the diagonals) of $\widehat{\Sigmam}_\bfy$ as
\begin{align}\label{Toepification}
[\widehat{\sigmam}]_k=\frac{\sum_{l=1}^{M-k+1} [\widehat{\Sigmam}_\bfy]_{l,l+k-1}}{M-k+1},
\end{align}
and where $\bfW$ is an $M \times M$ diagonal matrix with diagonal elements $(\sqrt{M}, \sqrt{2(M-1)}, \sqrt{2(M-2)},\dots, \sqrt{2})^\transp$ and takes into account the number of repetition of the elements in an $M \times M$ Hermitian Toeplitz matrix. Finally, by defining $\bfb=\bfW \widehat{\sigmam}-N_0[\bfW]_{1,1}\bfe_1$ and $\bfA=[\bfW \psim_1, \dots, \bfW \psim_G]$ we can write \eqref{eq:nnls_2} more compactly as the following \textit{Non-Negative Least-Squares} (NNLS) problem 
\begin{align}\label{eq:nnls_3}
\gammam^\star=\argmin_{\gammam \geq 0} \|\bfA \gammam - \bfb\|^2,
\end{align}
which can be efficiently solved with off-the-shelf optimization toolboxes (such as ``lsqnonneg.mat'' in MATLAB).
NNLS in \eqref{eq:nnls_3} has several interesting features 
\cite{slawski2013non, RN276} that it promotes the sparsity of the coefficients $\gammam$ without any need for additional 
sparsity-promoting regularizations such as $\ell_1$-norm traditionally used in CS algorithms such as the classical LASSO  \cite{tibshirani1996regression}.
Also, the past literature starting with Donoho {\em et al.}  \cite{donoho1992maximum} and more recent results  illustrate that non-negativity constraint alone  suffices to recover a sparse non-negative signal from under-determined linear 
measurements both in the noiseless  \cite{bruckstein2008uniqueness, donoho2010counting, wang2009conditions, wang2011unique}
and in the noisy  \cite{slawski2013non, RN276} case. 

It was shown in \cite{RN276} that a necessary condition  on the coefficient matrix $\bfA$ for NNLS to recover the sparse vector $\gammam$ efficiently  is that there exists a vector $\bft \in \bR^M$ such that $\Re[\bfA^\transp \bft] >0$. Interestingly, in  our case, this condition is immediately satisfied since the first row of the matrix $\bfA$ is given by 
\begin{align}
[\bfA]_{1,g}&=[\bfW]_{1,1} [\psim_g]_1=\sqrt{M} \int _{-1}^1 \psi_g(\xi) d \xi\\
& = \sqrt{M} \int _{-1}^1 R(\xi-\xi_g) d \xi\\
&= \sqrt{M} \int _{-1}^1 R(\xi) d \xi  >0,
\end{align}
where $R(\xi)$ is the rectangular pules of width $\frac{2}{G}$ and centered at zero, introduced before.
Hence, the necessary condition for NNLS is immediately satisfied by setting $\bft=(1,0,\dots, 0)^\transp$, which yields 
\begin{align}
\Re[\bfA^\transp \bft]= \sqrt{M} \int_{-1}^1 R(\xi) d \xi \times (1,1,\dots, 1)^\transp> 0. 
\end{align}

\begin{figure*}[t]
	\centering
	\includegraphics[scale=0.55]{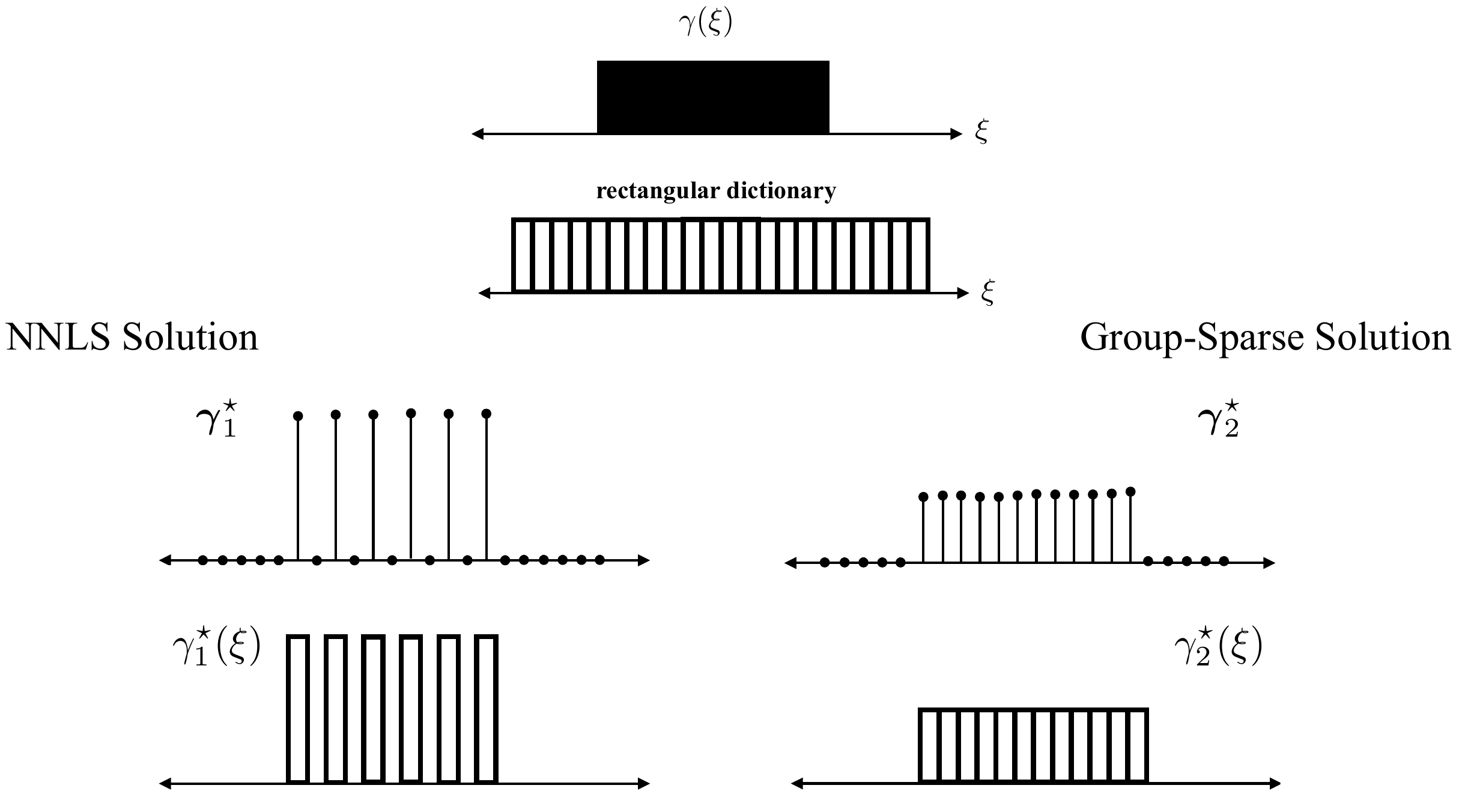}
	\caption{Sparsity promoting nature of NNLS and the lack of group-sparsity.}
	\label{fig:nnls_grouping}
\end{figure*}

\subsection{Geometrically Consistent ASF Estimation Using Generalized NNLS}\label{geom-consistent}
As we illustrate in our simulation results, NNLS indeed yields a very sparse solution for $\gammam$, thus, naturally favors ASFs that are sparse in the angular domain. However, NNLS does not necessarily yield group-sparse solutions. This problem is more evident when one increases the dictionary size $G$ much beyond the number of antennas $M$ to get a better approximation of the ASF since in that case the angular width of rectangular pulses $\psi_g(\xi)$, given by  $\frac{2}{G}$, goes much below the spatial resolution $O(\frac{1}{M})$  of the ULA. As a result, NNLS returns solutions that are highly spiky over the support of each scatterer, and are not necessarily group-sparse.

Let us explain this problem more in detail with  the simple example of rectangular ASF illustrated  in Fig.\,\ref{fig:rect_dict}.
This figure also illustrates two possible  ASF estimates. In the first,  $\gamma(\xi)$ is approximated by a sparse set of coefficients ${\gammam}^\star_1$, which yields the estimated ASF ${\gamma}^\star_1(\xi)$. It is important to note that although ${\gamma}^\star_1(\xi)$ is far from the original rectangular ASF $\gamma(\xi)$, it  yields almost the same covariance matrix as the original ASF $\gamma(\xi)$ when the rectangular pulses are sufficiently narrow (narrower than $\frac{1}{M}$ given by the angular resolution of the ULA). This is because
the mapping from ASFs into covariance matrices in \eqref{eq:asf_cov} is not one-to-one, and several ASFs apparently quite different from each other may yield very similar covariance matrices.  In the second case, we consider another approximation of the ASF where the set of coefficients ${\gammam}^\star_2$  is group-sparse but not as sparse as ${\gammam}^\star_1$  in the first case. 
Intuitively speaking (and as checked via numerical simulations), the proposed NNLS algorithm is more likely to produce a vector of coefficients with sparsity pattern as in ${\gammam}^\star_1$ than as in ${\gammam}^\star_2$, thus, it does not yield an ASF with a group-sparsity structure as in the original ASF $\gamma(\xi)$ (or ${\gamma}^\star_2(\xi)$).

Our goal in this section is to modify the NNLS algorithm in order to enforce the desired group-sparsity structure. 
We first define 
\begin{align}
\clP_p=\left \{\frac{\bfe_i+\dots+\bfe_{i+p-1}}{\sqrt{p}}: i\in [G-p]\right \},
\end{align}
where $\bfe_i$ denotes the $i$-th canonical vector in $\bR^G$ with $1$ as its $i$-th component and $0$ elsewhere. Note that for $p=1$, $\clP_p$ consists of all canonical vectors in $\bR^G$. In general, for any $p$, $\clP_p$ consists of all \textit{discrete rectangular pulses} of width $p$ and $\ell_2$ norm equal to $1$. Fig.\,\ref{fig:4_pulses} illustrates these pulses for $G=4$. 
\begin{figure}[t]
	\centering
	\includegraphics[scale=0.55]{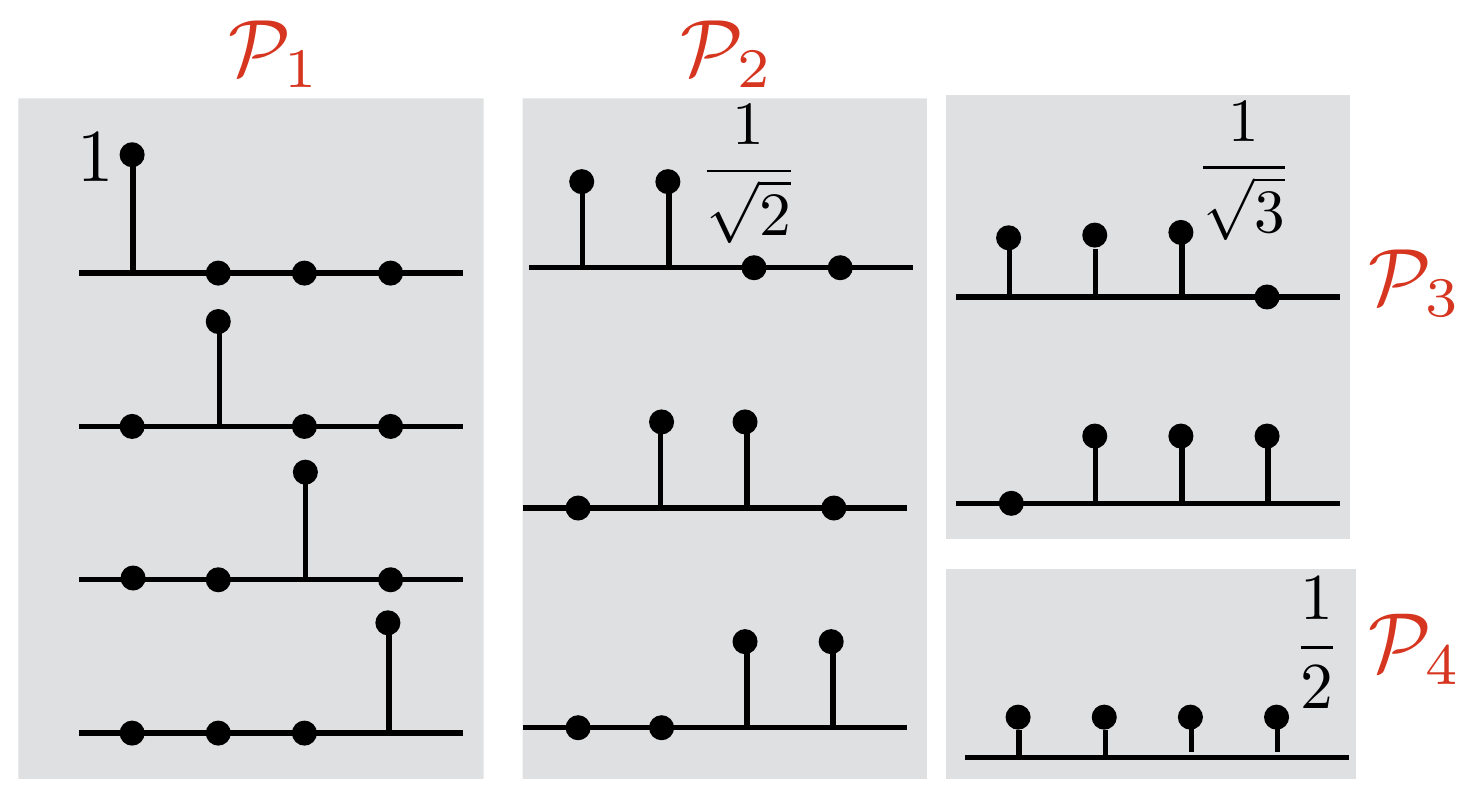}
	\caption{Collection of rectangular pulses $\clP_p$, for $p=1,2,3,4$.}
	\label{fig:4_pulses}
\end{figure}
We set a specific $p_0$ and define the discrete dictionary $\clD_{p_0}=\cup_{p=1}^{p_0} \clP_p$. 
Let $\gammam \in \bR^G$ be the vector of positive coefficients in the original NNLS problem in \eqref{eq:nnls_3}.
We claim that by enforcing the sparse representation of $\gammam$ over the dictionary $\clD_{p_0}$, for some $p_0\geq 2$, we can  promote the group-sparsity we aim to obtain. 
Let us fix a $p_0$ and let us define the size of the dictionary $\clD_{p_0}$ by  $D=|\clD|$ and the  $G\times D$ matrix consisting of the elements of $\clD$ by $\bfD$.
%
\def\alpham{\boldsymbol{\alpha}}
For any positive vector $\gammam\in \bR^G$, we define the sparse representation of $\gammam$ over $\clD_{p_0}$ by 
\begin{align}\label{atomic_norm}
\alpham(\gammam)=\argmin_{\alpham\in \bR_+^D} \|\alpham\|_1\text{ s.t. } \bfD \alpham=\gammam,
\end{align}
where we exploit the widely-adopted $\ell_1$-norm minimization to promote the sparsity of the representation $\alpham(\gammam)$. 
It is worthwhile to mention that in contrast with $\ell_1$-norm minimization for general sparse approximation, where the coefficients $\alpham$ in \eqref{atomic_norm} can be positive or negative, here we use only positive coefficients $\alpham$. The main reason is that   the original vector $\gammam$ and all the atoms of the dictionary have only positive coefficient and, as we will explain in the following, positivity of $\alpham$ imposes the group-sparsity structure we desire to have.

For $p_0=1$, the dictionary $\clD_{p_0=1}=\clP_1$ consists of only canonical vectors, thus, the sparsest representation of $\gammam$ over $\clD_{p_0=1}$ in \eqref{atomic_norm} will correspond to $\gammam$ itself, hence, no group-sparsity structure can be imposed. 
For $p_0\geq 2$, however, we can illustrate that sparse representation over $\clD_{p_0}$ favors group-sparse vectors $\gammam$. 

Let us show this by the following simple example.  Let us  set $G=4$ and consider the following two vectors $\gammam_1=(2,0,2,0)^\transp=2\bfe_1+ 2\bfe_3$ and $\gammam_2=(1,1,1,1)^\transp=\bfe_1+\dots+\bfe_4$ (see, e.g., Fig.\,\ref{fig:nnls_grouping}). For $p_0=1$, we can simply check that  both vectors $\gammam_1$ and $\gammam_2$ have the same $\ell_1$-norm $4$ over $\clD_{p_0=1}$.
Now let us consider $p_0=2$ and consider a larger dictionary $\clD_{p_0=2}=\clP_1 \cup \clP_2$ that  consists of all discrete rectangular pulses of width $1$ and $2$. A simple calculation following \eqref{atomic_norm} shows that the sparse representation of $\gammam_1$ in $\clD$ is given by 
\begin{align}
\gammam_1=(2,0,2,0)^\transp=2\times (1,0,0,0)^\transp + 2\times (0,0,1,0)^\transp.\nonumber
\end{align}
This is because $\gammam_1=(2,0,2,0)^\transp$ has a $0$ between its two non-zero elements, thus, it cannot contain any discrete pulse of width $2$ since each such pulse consists of two adjacent non-zero values. 
This implies that  $\ell_1$-norm of $\gammam_1$ is equal to $4$, thus, the same as its $\ell_1$-norm over the smaller dictionary $\clD_{p_0=1}$. 

For $\gammam_2=(1,1,1,1)$, in contrast, we can obtain the following sparse representation as the linear combination of two rectangular pules of width $2$:
\begin{align}
\gammam_2 =(1,1,1,1)^\transp&= \sqrt{2} \times \big (\frac{1}{\sqrt{2}}, \frac{1}{\sqrt{2}}, 0,0\big)^\transp \nonumber\\
&+ \sqrt{2} \times \big(0,0,\frac{1}{\sqrt{2}}, \frac{1}{\sqrt{2}}\big)^\transp.
\end{align}
It is seen that the resulting sparse representation has the $\ell_1$-norm  $2 \sqrt{2}$, which is lower than the $\ell_1$-nomr of $\gammam_1$ over the smaller dictionary $\clD_{p_0=1}$. This simple example clearly illustrates  that promoting the sparsity of  the vector of coefficients $\gammam$ over $\clD_{p_0}$ for some $p_0\geq 2$ favors those $\gammam$ having  the group-sparsity structure.

Based on this simple observation we can alleviate the issue caused by NNLS (see, e.g., Fig.\,\ref{fig:nnls_grouping}) as follows.
We set a number $p_0\geq 2$ and define the corresponding dictionary $\clD_{p_0}$ and corresponding dictionary matrix $\bfD$. We modify the previous NNLS algorithms where instead of solving \eqref{eq:nnls_3} as
\begin{align}\label{NNLS_mod}
{\gammam}^\star=\argmin _{\gammam\geq 0} \|\bfA \gammam - \bfb \|^2,
\end{align}
 we solve the following optimization after incorporating \eqref{atomic_norm} 
\begin{align}\label{LS_alg_group}
{\alpham}^\star=\argmin _{\alpham\geq 0} \| \bfA \bfD \alpham - \bfb\|^2 + \varsigma \|\alpham\|_1
\end{align}
where $\varsigma>0$ is a regularization parameter, which together with $\ell_1$-norm regularization on $\alpham$ promotes the sparsity of $\alpham$, and 
where afterwards we estimate ${\gammam}^\star$ as ${\gammam}^\star=\bfD {\alpham}^\star$. 
It is important to note that since all the canonical vectors in the set $\clP_1$ are included in the dictionary $\clD_{p_0}$, for all $p_0=1,2,\dots$, for $\varsigma=0$, the NNLS \eqref{LS_alg_group} yields the solution of the original NNLS \eqref{eq:nnls_3}. This can be seen simply by setting $\alpham^\star$ equal to $\lambdam^\star$ at the coordinates corresponding to $\clP_1$ and zero elsewhere. Therefore, by varying $\varsigma \in [0, \infty]$ we obtain a collection of ASFs with more and more group-sparsity. 

Although optimization problem \eqref{LS_alg_group} has the additional $\ell_1$-norm regularization, the following proposition shows that, due to the non-negativity of $\alpham$, it can be still posed as an NNLS, which we call generalized NNLS in the following.
\begin{proposition}\label{nnls_equiv}
	Let ${\alpham}^\star$ be the optimal solution of \eqref{LS_alg_group} and suppose that ${\alpham}^\star\not = 0$. Then, ${\alpham}^\star$ is the optimal solution of 
	\begin{align}\label{LS_alg_group2}
	{\alpham}^\star=\argmin _{\alpham\geq 0} \| \bfA \bfD \alpham - \bfb\|^2 + \varsigma' \|\alpham\|_1^2,
	\end{align}
	provided that $\varsigma'=\frac{\varsigma}{2\|{\alpham}^\star\|_1}$. \hfill $\square$
\end{proposition}
\begin{proof}
	For simplicity, let us define $\Phim=\bfA \bfD$. Then, we can write the KKT conditions \cite{boyd2004convex} for the optimizations as
	\begin{align}
	&2\,\Re[\phim_i^\herm(\Phim \alpham - \bfb)] + \varsigma + \beta_i=0, & \beta_i \alpha_i=0, \alpha_i \geq 0,\label{LS_alg_group_KKT}\\
	&2\,\Re[\phim_i^\herm(\Phim \alpham - \bfb)] + 2\varsigma' \|\alpham\|_1 + \beta'_i=0, &  \beta'_i \alpha_i=0, \alpha_i \geq 0,\label{LS_alg_group2_KKT}
	\end{align}
	where $\phim_i$ denotes the $i$-th column of $\Phim$ and where $\beta_i\geq 0$ and $\beta_i'\geq 0$ are the KKT coefficients corresponding to the positivity of $\alpham$ in \eqref{LS_alg_group} and \eqref{LS_alg_group2}, respectively. Since $\alpham^\star$ satisfies \eqref{LS_alg_group}, one can see that if $\alpham^\star\not = 0$ it also satisfies the KKT condition in \eqref{LS_alg_group2_KKT} by setting $\beta_i'=\beta_i$ and $\varsigma'=\frac{\varsigma}{2\|\alpham^\star\|_1}$. Since \eqref{LS_alg_group2} is a convex optimization problem, this implies that $\alpham^\star$ is also the optimal solution of  \eqref{LS_alg_group2} for the parameter $\varsigma'=\frac{\varsigma}{2\|\alpham^\star\|_1}$. 
\end{proof}
Using Proposition \eqref{nnls_equiv}, we can solve \eqref{LS_alg_group2} rather than the $\ell_1$-norm regularized function in \eqref{LS_alg_group}. In particular, by introducing 
\begin{align}
\wch{\bfA}=\left [ \begin{matrix} \sqrt{2 \varsigma'} \times {\bf 1}^\transp \\ \bfA \bfD \end{matrix}\right ], \wch{\bfb}=\left [ \begin{matrix}0  \\ \bfb \end{matrix}\right ],
\end{align}
where ${\bf 1}= (1, \dots, 1)^\transp$ denotes the all-1 vector, we can write \eqref{LS_alg_group2} more compactly as the following NNLS
\begin{align}\label{LS_alg_group3}
{\alpham}^\star=\argmin _{\alpham\geq 0} \| \wch{\bfA}  \alpham - \wch{\bfb}\|^2,
\end{align}
which similarly to \eqref{eq:nnls_3} can be efficiently solved with off-the-shelf optimization toolboxes.

\noindent{\bf Criterion for choosing $p_0$.}
It is  worthwhile here to mention that the parameter $p_0$ controls the width of the group, that is, larger $p_0$ favors larger groups inside the support. For example, as we explained, for $p_0=1$ the dictionary $\clD_{p_0=1}$ is unable to force any group-sparsity structure in the sense that the vectors $\gammam_1=(2,0,2,0)^\transp$ and $\gammam_2=(1,1,1,1)^\transp$, the former with a smaller group size and the latter with a larger group size, have the same $\ell_1$-norms. The $\ell_1$-norms,  however, change  for $\clD_{p_0=2}$. It is not difficult to check that similarly $p_0=2$ is ineffective to capture group-sparsity for groups of size larger than $2$. For example, two vectors $\gammam_3=(1,1,1,1,0)^\transp$ and $\gammam_4=(1,1,0,1,1)^\transp$ have the same $\ell_1$-norm of $2\sqrt{2}$ over $\clD_{p_0=2}$ although the former has a better group-sparsity than  the latter. To incorporate this and promote larger group sizes inside the support, we need to increase $p_0$. For example, setting $p_0=3$ and following similar steps,  we can show that the $\ell_1$-norm of $\gammam_3$  and $\gammam_4$ in the  dictionary $\clD_{p_0=3}$ would be $1+\sqrt{3}=2.7321$ and $2\sqrt{2}=2.8284$. Therefore, it is seen that $\gammam_3$, which has a better group-sparsity than $\gammam_4$ is favored in this new dictionary. In theory, by increasing $p_0$ one gets better group-sparsity at the cost of computational complexity since the size of dictionary  $\clD_{p_0}$, thus, the dimension of the vector $\alpham$, grows proportionally to $p_0$ as  
\begin{align}
\sum_{i=1}^{p_0} G-i+1=O(p_0 G),
\end{align}
and ultimately approaching $O(G^2)$ for a full dictionary $\clD_{p_0=G}$. In practice, as we investigated with simulation results, selecting $p_0$ larger than $\frac{G}{M}$ does not give a significant improvement. Intuitively speaking, this is due to the fact that, the lack of group-sparsity (see, e.g., Fig.\,\ref{fig:nnls_grouping}) emerges when $G$ is very larger than $M$ such that the width of the rectangular pulses is much below the angular resolution $\frac{2}{M}$ of the array. In those cases, the NNLS introduces a gap in  the support of the group by creating zero elements in between. By setting $p_0$ to $\frac{G}{M}$ or larger, we make sure that this gap in support will not happen over any AoA interval of width smaller than the resolution of the antenna $\frac{2}{M}$. This seems to be enough to alleviate the issue caused by NNLS.  Our simulation results fully confirm this intuition.

\section{A Machine Learning Approach Using Deep Neural Networks  for group-sparsity}
Although the NNLS together with the group-sparsity regularization imposed through the dictionary  $\clD_{p_0}$, for some $p_0\geq 2$, seem to be an efficient approach to obtain group-sparse solution, it is not scalable in terms of complexity when the number of antennas and the size of rectangular pulsed grid $G$ is quite large. Moreover, solving  NNLS \eqref{eq:nnls_3} and regularized NNLS \eqref{eq:nnls_3} through iterative techniques in quite challenging in many implementation of massive MIMO. In this section, we propose a Machine Learning (ML) approach using \textit{Deep Neural Networks} (DNNs) that tries to learn the notion of \textit{``group-sparsity''} from training data. More specifically, rather than applying any regularization to promote group-sparsity as we did in the previous approach, we let the DNN learn and ultimately extract group-sparse solutions from the observation of the i.i.d. channel vectors.
To be more specific, we consider the following steps.

\noindent{\bf a.\,Training Data.} We produce many ASFs with group-sparse structure as in \eqref{asf_group} with different number of groups $K\in \{1,2,3,4\}$. Moreover, we assume that the continuous function $\gamma_k$ corresponding to the group $k$, $k \in [K]$, is a pulse with a connected support of width  chosen  uniformly randomly in the interval $[0,0.3]$, namely, we assume that the scatterers in the propagation channels have an angular width of at most $0.3$. 
For each ASF $\gamma$ inside this set, we compute the corresponding Toeplitz covariance matrix and extract the $M$-dim vector $\sigmam$ corresponding to the first column of the covariance matrix. Recall that due to the Toeplitz structure, $\sigmam$ contains all information about the covariance matrix. Then, we produce a  noisy versions of $\sigmam$, denotes by $\wch{\sigmam}$, and add the $(\wch{\sigmam}, \gammam)$ as labeled sample to the training data, where $\gammam$ denotes a discrete quantization (sampling)  of $\gamma(\xi)$ over a uniform grid of size $G$.
By repeating this for all the ASFs inside the group, we produce our training data.
Note that we call $(\wch{\sigmam}, \gammam)$ a labeled sample since we hope that after suitable training DNN be able to recover $\gammam$ given the noisy samples $\wch{\sigmam}$.

\noindent{\bf b.\,Supervised Learning.} We consider a supervised training using the labeled data $\{(\wch{\sigmam}_s, \gammam_s): s \in [S]\}$ where $S$ denotes the number of training samples. We use these training samples to train a \textit{Deep Neural Networks} (DNN). For training, we use  the widely-adopted \textit{Stochastic Gradient Descend} (SGD) with $\ell(\gammam, \widehat{\gammam})=\|\gammam - \widehat{\gammam}\|_1$ as the loss function between the true $\gammam$ and the estimate $\widehat{\gammam}$ generated by the network.

\begin{figure}
	\centering
	\includegraphics[scale=0.5]{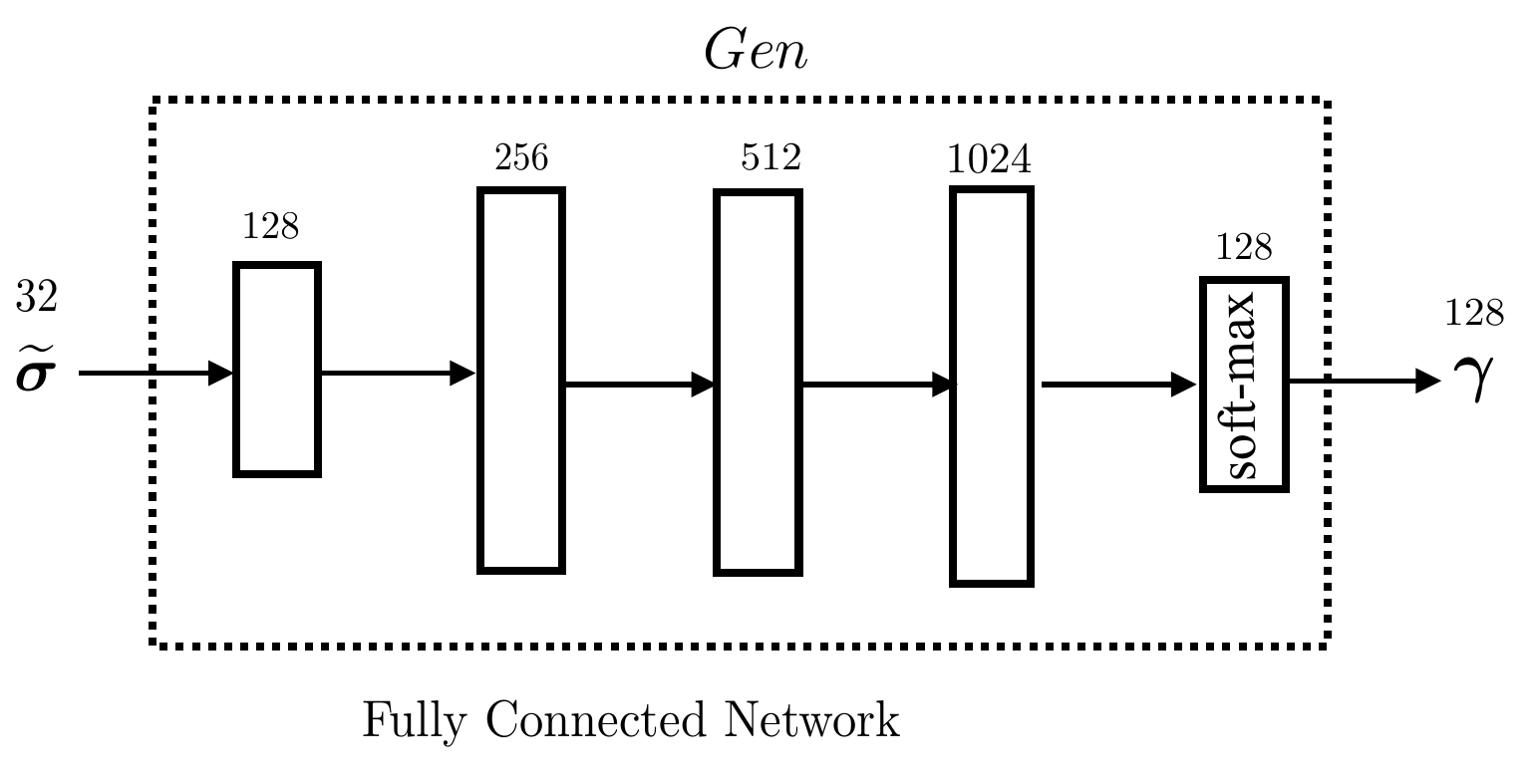}
	\caption{The structure of DNN adopted for ASF estimation. DNN consists of $5$ layers with $128$, $256$, $512$,  $1024$, and $128$ neurons, respectively. The last layer has a  \textit{soft-max} activation function and produces positive values for $\gammam\in \bR_+^128$. }
	\label{DNN}
\end{figure}


\noindent{\bf c.\,Structure of DNN.}
One of the important factors affecting the performance of the ASF estimation using DNNs is the structure of the DNN consisting of the number of layers, the number of neurons in each layer, and  the activation function of each layer. In this paper, we use a fully-connected network illustrated in Fig.\,\ref{DNN}, with $5$ layers consisting of $128$, $256$, $512$, $1024$, and $128$ neurons, respectively, where the number of neurons $128$ in the last layer corresponds to the grid size $G=128$ we are adopting for ASF quantization. The activation function of the $4$ initial layers is the RelU function $x \mapsto \max\{x,0\}$. For the last layer we use the \textit{soft-max} activation function, which for an input vectors $(x_1, \dots, x_G) ^\transp$ in the input produces the output as $(x_1, \dots, x_{G})^\transp \mapsto \frac{(e^{x_1}, \dots, e^{x_{G}})^\transp}{\sum_{j=1}^{G} e^{x_j}}$ where $G=128$, denotes the grid size we use for the quantization of ASF. Note that the summation of the elements produced by soft-max layer is always $1$, which produces a normalized $\gammam$, i.e., $\sum_{i=1}^G\gamma_i=1$.

\begin{figure*}[t]
	\centering
	\begin{subfigure}[t]{0.3\textwidth}
		\includegraphics[width=\textwidth]{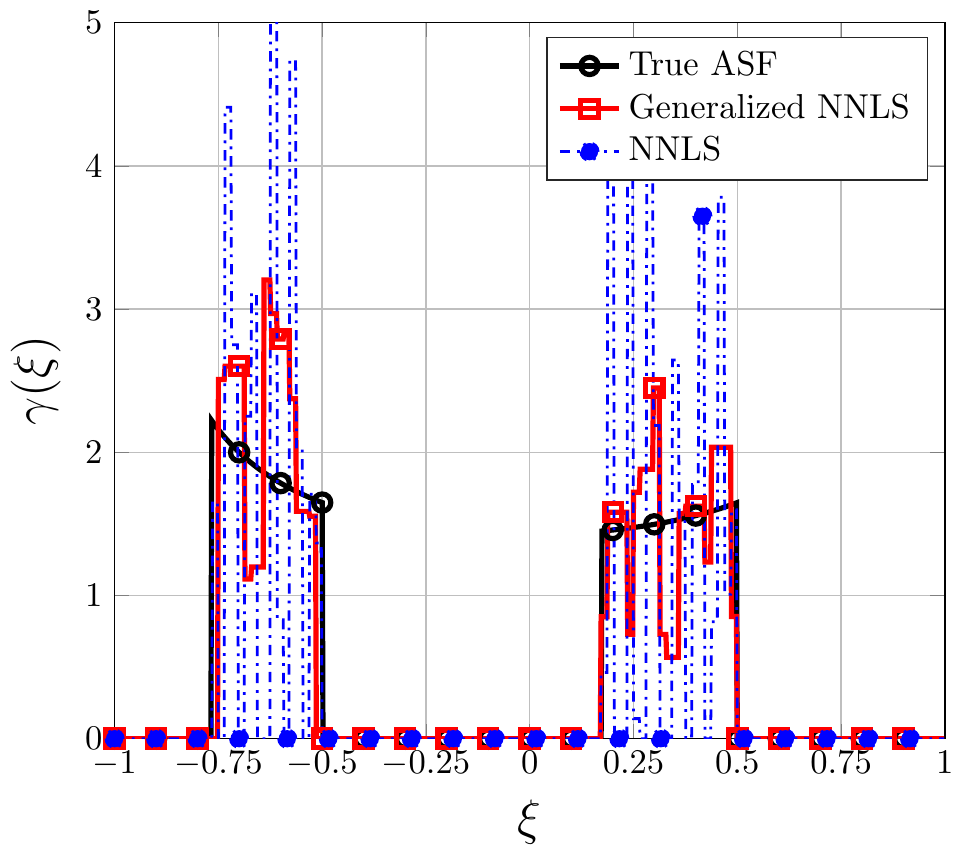}
		\caption{$\frac{T}{M}=2$}
	\end{subfigure}
~
	\begin{subfigure}[t]{0.3\textwidth}
		\includegraphics[width=\textwidth]{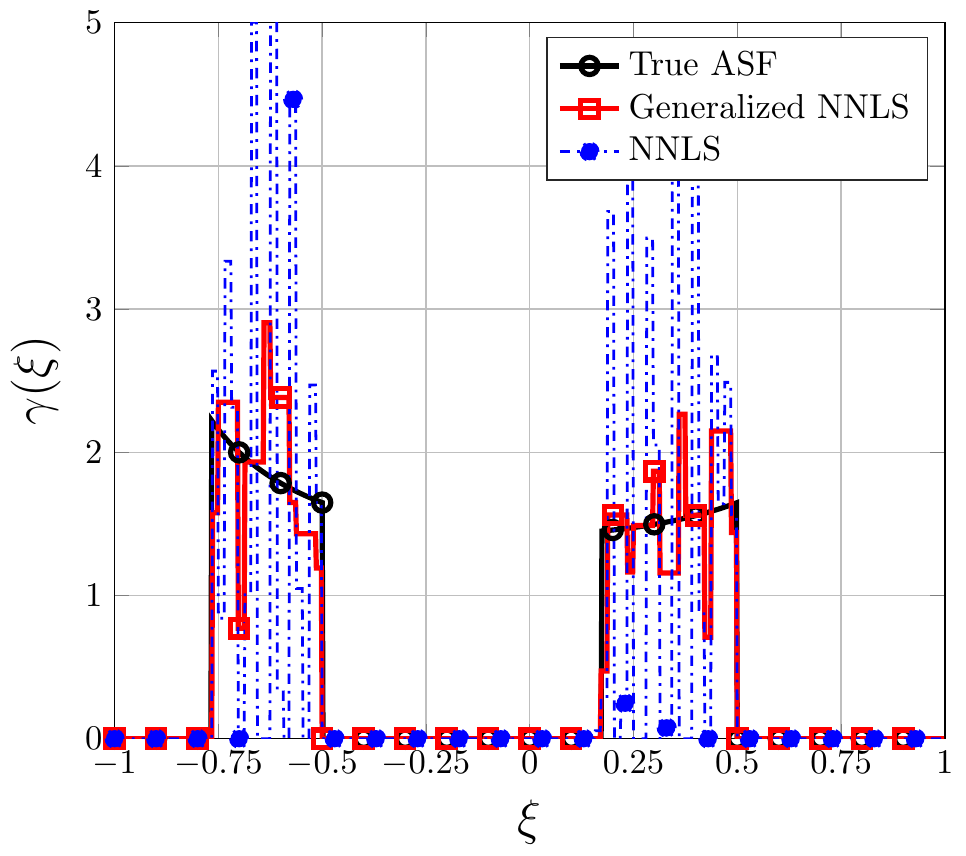}
		\caption{$\frac{T}{M}=4$}
	\end{subfigure}
~
\begin{subfigure}[t]{0.3\textwidth}
	\includegraphics[width=\textwidth]{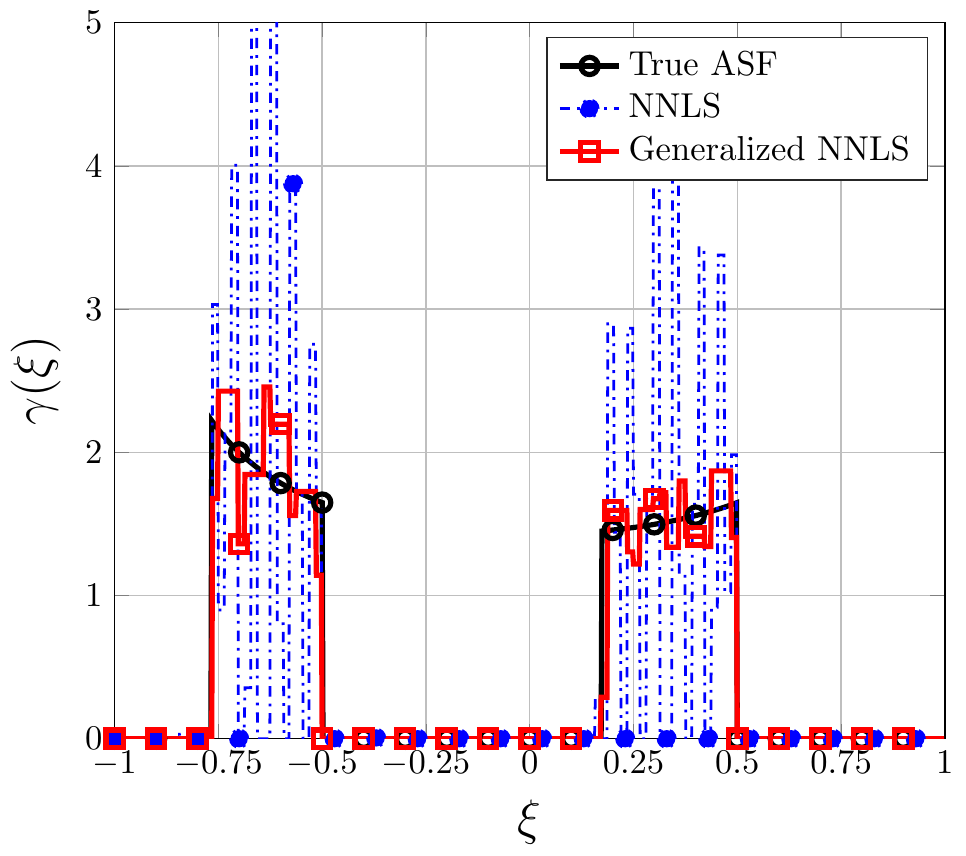}
	\caption{$\frac{T}{M}=8$}
\end{subfigure}

	\begin{subfigure}[t]{0.3\textwidth}
	\includegraphics[width=\textwidth]{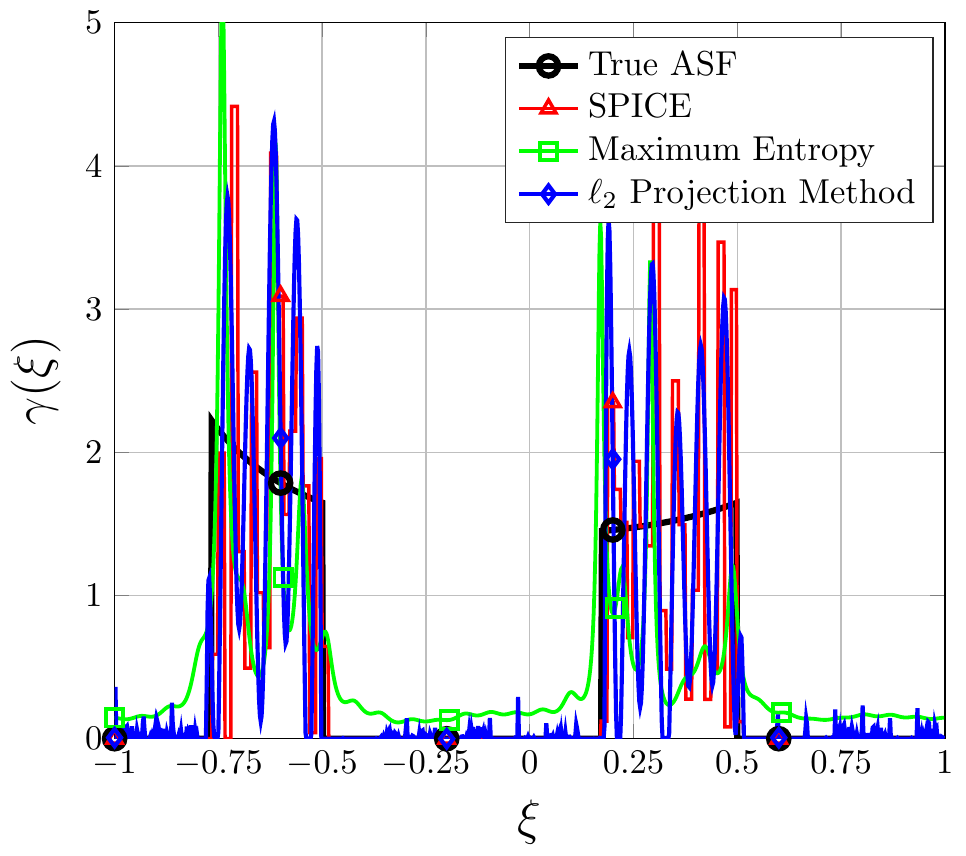}
	\caption{$\frac{T}{M}=2$}
\end{subfigure}
~
\begin{subfigure}[t]{0.3\textwidth}
	\includegraphics[width=\textwidth]{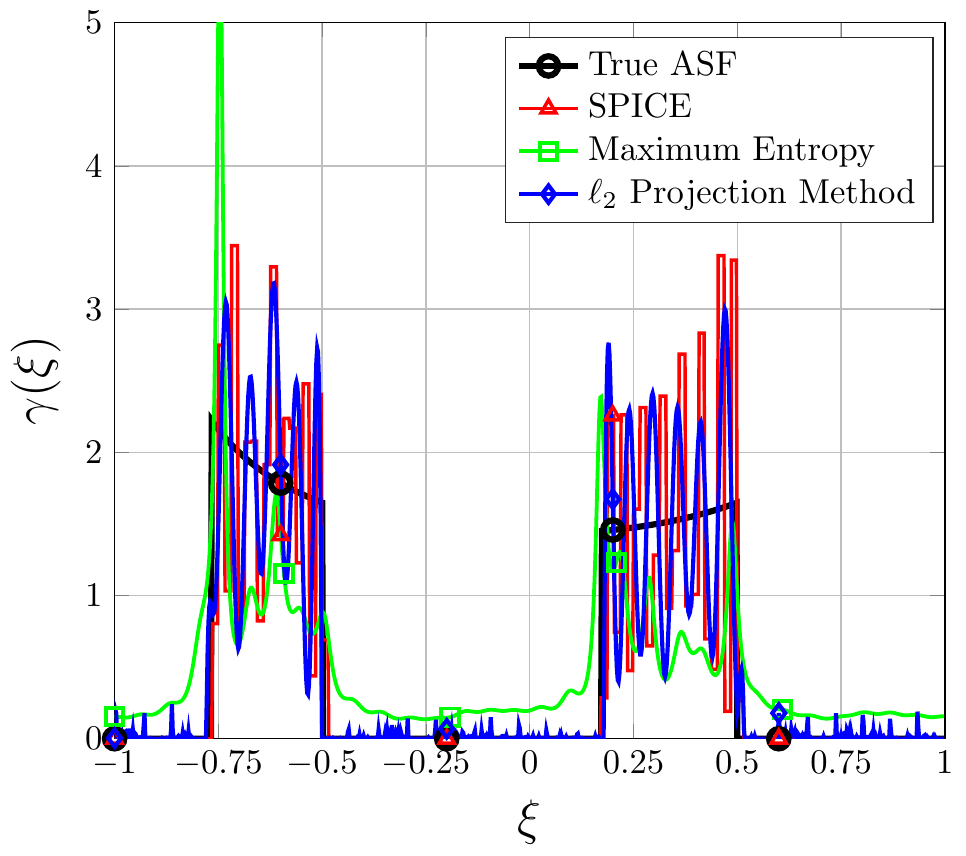}
	\caption{$\frac{T}{M}=4$}
\end{subfigure}
~
\begin{subfigure}[t]{0.3\textwidth}
	\includegraphics[width=\textwidth]{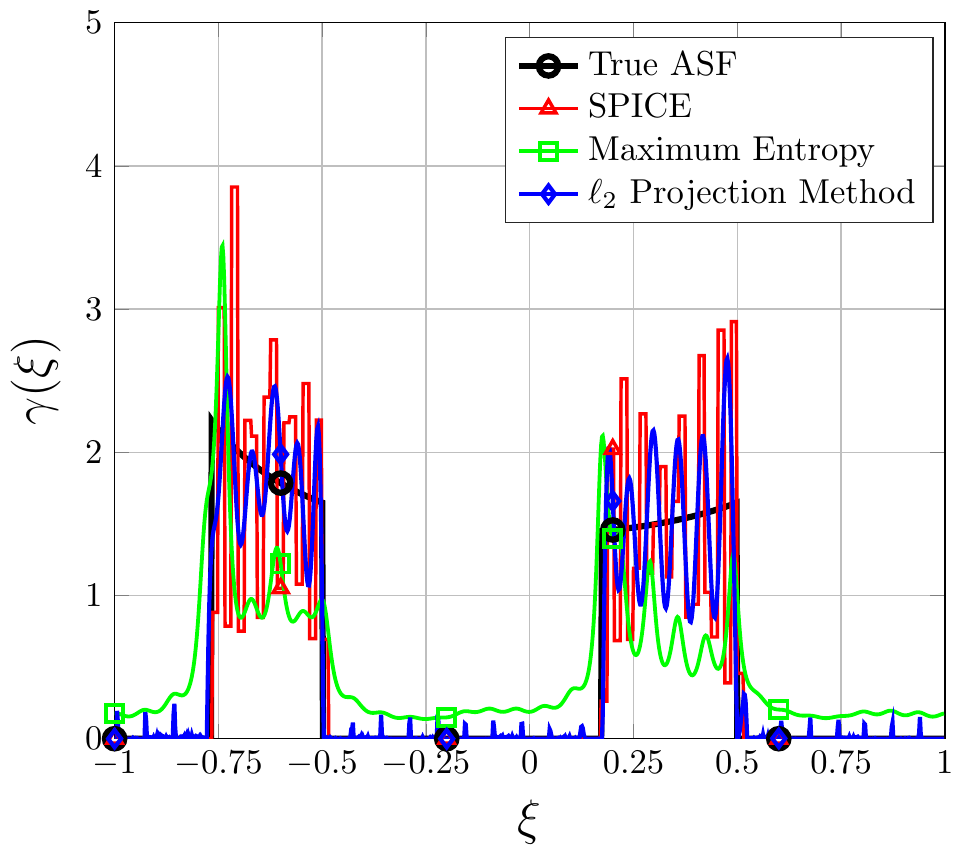}
	\caption{$\frac{T}{M}=8$}
\end{subfigure}

\begin{subfigure}[t]{0.3\textwidth}
	\includegraphics[width=\textwidth]{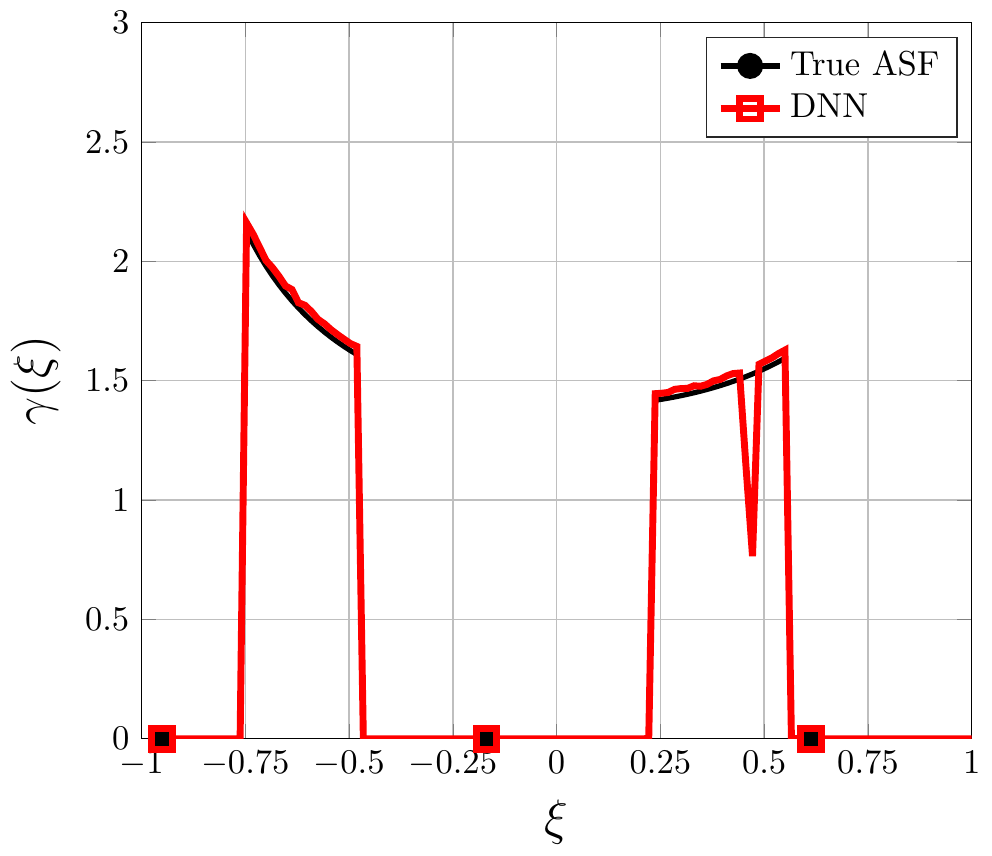}
	\caption{$\frac{T}{M}=2$}
\end{subfigure}
~
\begin{subfigure}[t]{0.3\textwidth}
	\includegraphics[width=\textwidth]{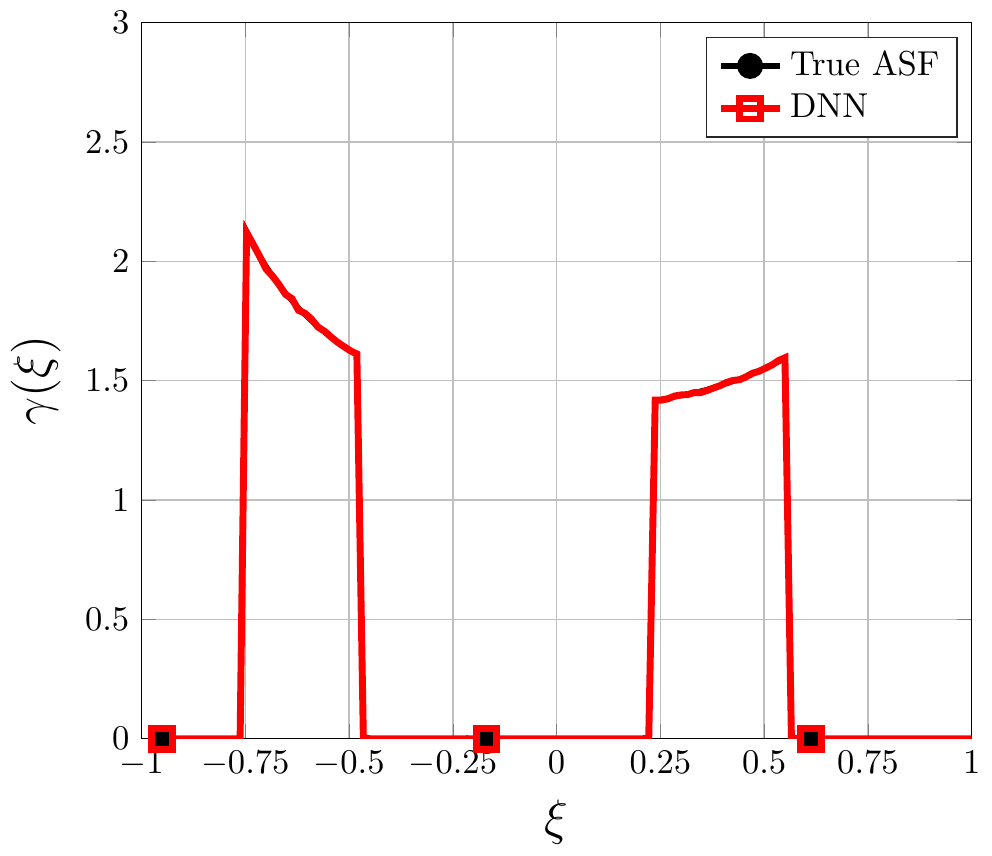}
	\caption{$\frac{T}{M}=4$}
\end{subfigure}
~
\begin{subfigure}[t]{0.3\textwidth}
	\includegraphics[width=\textwidth]{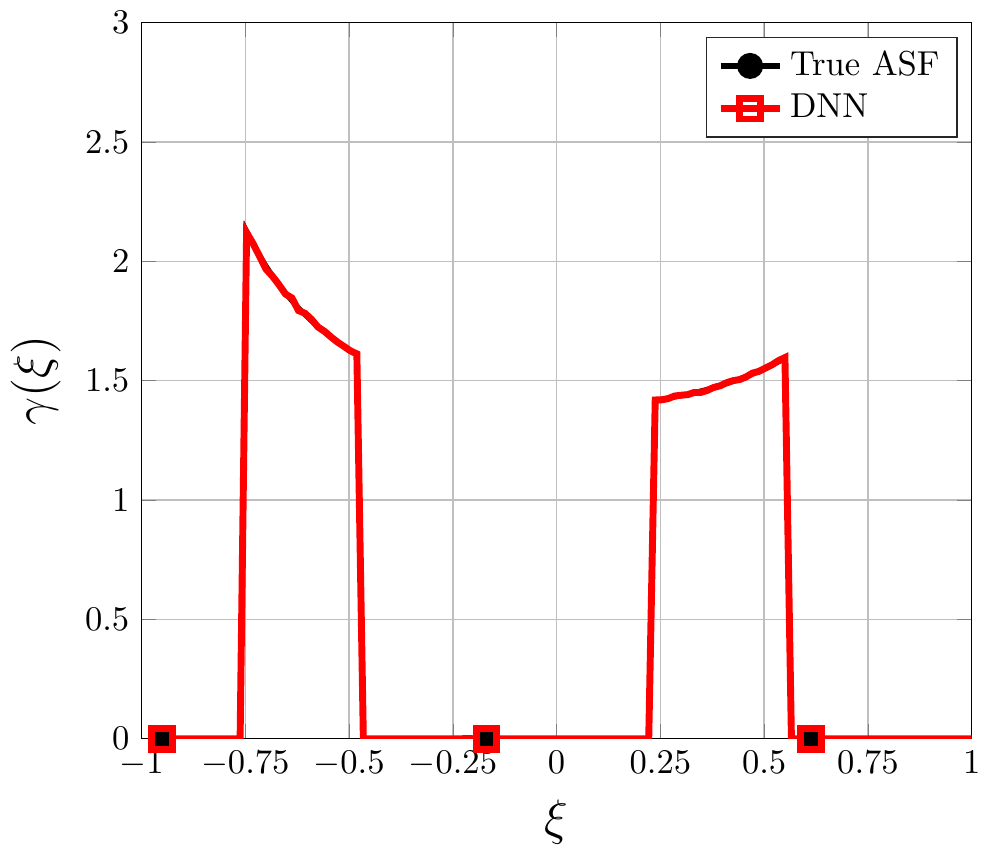}
	\caption{$\frac{T}{M}=8$}
\end{subfigure}
	\caption{ASF estimation comparison for various methods: the upper row illustrates the estimation for NNLS and generalized NNLS methods, and the lower row represents the estimation for SPICE, maximum entropy, and $\ell_2$-projection method. The number of noisy channel samples varies from left to right as: $T=2M$, $T=4M$, and $T=8M$.}
	\label{fig:compare_liter}
\end{figure*}

\section{Simulation Results}
In this section, we perform numerical simulations to assess the performance of our proposed algorithm. We compare our method with the following algorithms: (a) SPICE \cite{stoica2011spice}, (b) the well-known Burg's Maximum Entropy (BME) method for spectral estimation \cite{burg1967maximum}, and the recently-proposed $\ell_2$-norm projection method \cite{miretti2018fdd} given by \eqref{HHI_method}. For all the simulations, we consider a \textit{Signal-to-Noise Ratio} of 20\,dB for noisy channel vectors.

Fig.\,\ref{fig:compare_liter} illustrates the simulation results for different sampling ratios (number of samples per signal dim or number of antennas) $\frac{T}{M}$.
It is seen that the SPICE and also our proposed NNLS methods produce very sparse solutions but are not able to capture the group-sparsity structure.

 Maximum Entropy method and $\ell_2$-norm projection, in contrast, produce almost group-sparse estimates although both are quite fluctuating over the support of the true ASF, and have quite large out-of-band components off the support. 
This is typically the main problem of  classical power spectral estimation methods and also $\ell_2$-norm projection  as they produce large ripples when the ASF has sharp transitions, which is  the case with group-sparse ASFs we consider in this paper. This ripple-effect  can be evidently seen from the simulation results.

Our proposed generalized NNLS, however, performs quite well: it estimates the support almost perfectly, and reproduces the  amplitude of the ASF over the support  quite precisely (although not perfectly)   for large $\frac{T}{M}$. Also, compared with NNLS and SPICE, it is able to capture the group-sparsity very well.

 It is seen that DNN has an excellent performance much superior to that of other methods and also our proposed generalized NNLS. Note that we had trained DNN with $K\in \{1,2,3,4\}$ clusters and the ASF we use for testing in this simulation has $K=2$ clusters, thus, it belongs to the category of ASFs over which the DNN has been trained. And, it is seen that DNN is able to recognize the number of clusters of the ASF  perfectly for $\frac{T}{M}=4,8$. 
 Although trivial, but we would like to emphasize that the ASF chosen for testing has been selected completely randomly, and in particular it does not belong to the training set. Over all, one can see that DNN is able to estimate the number of clusters, their support, and the amplitude of ASF over the support almost perfectly.


	\balance
	
	{\small
		\bibliographystyle{IEEEtran}
		\bibliography{references2}
	}
	
\end{document}